\documentclass{easychair}
\usepackage{doc}
\usepackage{makeidx}
\usepackage{amsfonts,amssymb,amsmath,alltt}
\usepackage {stmaryrd}
\usepackage{graphicx}
\usepackage[utf8]{inputenc}
\usepackage{url}
\usepackage[english]{babel}
\usepackage{tikz}
\usetikzlibrary{arrows,calc,automata}

\usepackage{listings}

\newcommand{\listbox}[1]{\lstset{basicstyle=\footnotesize}{\mbox{\lstinline!#1!}}\lstset{basicstyle=\scriptsize,language=JAVA}}

\newenvironment{ttbox}{\begin{alltt}\ttbraces\small\tt}%
                      {\end{alltt}}
\def\ttbraces{\let\.=\nobreak\chardef\{=`\{\chardef\}=`\}\chardef\|=`\\}

\newcommand\ttOmega{\mbox{\( \Omega \)}}

\newcommand\imp\Rightarrow

\newcommand\ttand{\mbox{{$\land$}}}

\newcommand\ttfun{\mbox{{$\Rightarrow$}}}

\newcommand\ttimp{\mbox{{$\longrightarrow$}}}
\newcommand\ttequiv{\mbox{{$\equiv$}}}

\newcommand\ttforall{\mbox{{$\forall$}}}

\newcommand\ttin{\mbox{{$\in$}}}

\newcommand{\evs}{{\mathcal S}}
\newcommand{\ttevs}{\mbox{\(\evs\)}}
\newcommand{\pow}{{\mathcal P}}
\newcommand{\ttgeq}{\mbox{\(\geq\)}}
\newcommand{\ttcup}{\mbox{\(\cup\)}}

\newcommand{\ttcap}{\mbox{\(\cap\)}}

\newcommand{\ttvarnothing}{\mbox{\(\varnothing\)}}

\newcommand{\ttlam}{\mbox{\(\lambda\)}}
\newtheorem{theorem}{Theorem}
\newtheorem{definition}{Definition}[section]

\begin{document}
\title{QKD in Isabelle -- Bayesian Calculation}

\titlerunning{QKD in Isabelle}

\author{Florian Kamm\"uller\\
Middlesex University London\\
\url{f.kammueller@mdx.ac.uk}\\
%\and
}

\authorrunning{F. Kamm\"uller}

\maketitle
\begin{abstract}
In this paper, we present a first step towards a formalisation of
the Quantum Key Distribution algorithm in Isabelle. We focus on the 
formalisation of the main probabilistic argument why Bob cannot be certain 
about the key bit sent by Alice before he doesn't have the 
chance to compare the chosen polarization scheme. This means that any 
adversary Eve is in the same position as Bob and cannot be certain
about the transmitted keybits.
We introduce the necessary basic probability theory, present a graphical 
depiction of the protocol steps and their probabilities, and finally how this
is translated into a formal proof of the security argument.
\end{abstract}

\section{Introduction}
In this paper, we present a simple finite foundation for a formalisation of
parts of the Quantum Key Distribution (QKD) algorithm in Isabelle. We focus on the 
formalisation of the probabilistic theory needed to formalise and 
then calculate the probability with which Bob receives the key bit sent 
by Alice. This basic probability argument is important as it can also be
applied to the attacker Eve that has the same chance to learn the transmitted keybits.
Bob, however, in a second step compares the used polarisation schemes
with Alice. Thereby he and Alice are able to retain only key bits that 
have been correctly transmitted. 

To give a brief idea of the protocol that can be used to transmit a sequence of random
bits (that can then be used as a shared One-Time-Pad key giving 100\% security):
\begin{itemize}
\item[(1)] Alice randomly selects a bit 0 or 1 
\item[(2)] Alice randomly chooses diagonal 
(X) or rectilinear (+) polarisation schemes to encode the bit as a photon 
before sending the bit 
\item[(3)] Bob also randomly chooses schemes (X/+) before measuring
the received photon. According to quantum properties, if Alice and Bob chose the same
polarisation schemes the transmission is 100\% correct, if they use different ones
the changes are 50/50.
\end{itemize}
A representative list of possible combinations is given in Figure \ref{fig:qkd}.

\begin{figure}
\begin{center}
Alice sends\\
 \includegraphics[scale=.3]{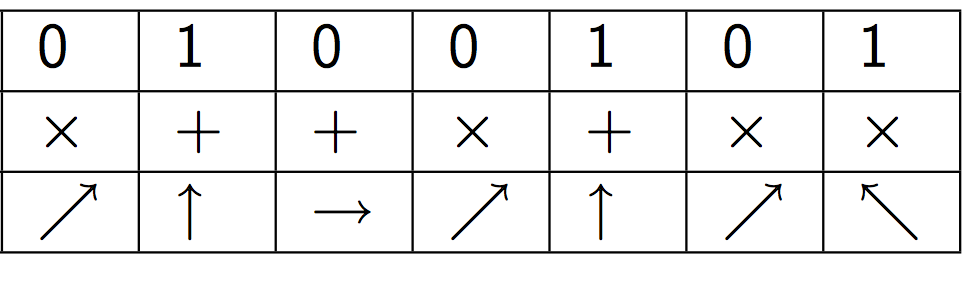}\\
Bob measures\\
 \includegraphics[scale=.3]{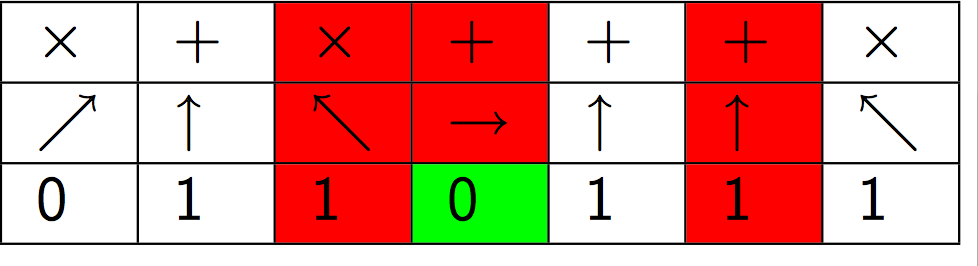}\\
they can use\\
 \includegraphics[scale=.3]{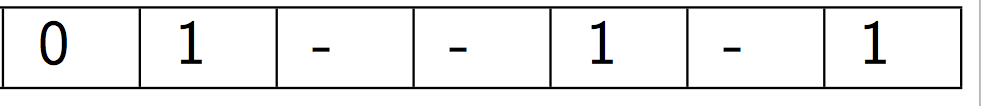}
\end{center}
\caption{QKD example (images from \cite{sin:99})}\label{fig:qkd}
\end{figure}

In this paper, we (1) introduce the necessary probability theory to
show (2) the basic probabilities of the correctness of the key transmission
in the protocol which is a step towards the security analysis and (3)
we develop and illustrate the probability reasoning on finite sets of outcomes
in Isabelle.
Note, that we consider only one bit since the principle is the same in
any number of repetitions necessary to transmit a $n$-bit key.
Also we only consider for a start the first phase of the protocol.

We introduce the necessary basic probability theory, present a graphical 
depiction of the protocol steps and their probabilities, and finally show how this
is translated into a formal proof of the security argument.
The theory and all presented proofs are formalised in Isabelle (see Appendix).

\subsection{Basic Probability}
The very brief introduction to basic probability theory is taken from 
Koller and Friedmann \cite{kf:09} but vastly abbreviated. The 
reader is referred to this excellent textbook for details.

Before defining events $(\evs)$ we first assume a set $\Omega$ of possible
outcomes. Based on that we define a set of {\it measurable} events 
$\evs \subseteq \pow \Omega$ where $\pow$ is the power set. 
Any event $A \in \evs$ may have probabilities assigned to it. 
Probability theory, more precisely, measure theory 
(see \cite{hur:03}), requires that the the following conditions hold for 
the probability space $\evs$:
\begin{itemize}
\item $\evs$ contains the empty event $\varnothing$ and the trivial event $\Omega$;
\item $\evs$ must be closed under union: $A, B \in evs \ttfun A \cup B \in evs$;
\item $\evs$ must be closed under complement: $A \in evs \ttfun \Omega \setminus A \in evs$.
\end{itemize}
The closure for the other Boolean operators intersection and set difference is implied
by the above conditions.

\begin{definition}[Probability Distribution]\label{def:dist}
A {\it probability distribution} $P$ over $(\Omega, \evs)$ is a function
from events in $\evs$ to real numbers satisfying the following conditions.
\begin{enumerate}
\item $\forall A \in \evs.\  P(A) \geq 0$. 
\item $P(\Omega) = 1$.
\item If $A, B \in \evs$ and $A \cap B = \varnothing$ then 
$P(A \cup B) = P(A) + P(B)$.
%\item More generally, $P(A \cup B) = P(A) + P(B) - P(A \cap B)$.
\end{enumerate}
\end{definition}
In Joe Hurd's dissertation \cite{hur:03} these conditions are referred to as 
\begin{itemize}
\item[(1)] {\it Positivity}, 
\item[(2)] {\it Probability space} ((2.27), page 33 \cite{hur:03}), and 
\item[(3)] {\it Additivity} 
\end{itemize}
in the general context of Measure spaces. 
The property of Monotonocity and Countable Additivity \cite{hur:03}
are not present in the introduction of Koller and Friedman but at least 
Countable Additivity can be considered as implicit since we are looking 
at finite spaces only.

The above definition can be directly translated into an Isabelle specification\footnote{Even
though measure theory \`a la Hurd is provided in the Isabelle theory library, we prefer
to provide a simpler ad hoc definition here for completeness -- integration is possible
and planned for later stages.}.
We transform the textbook definition into a definition and a 
type definition: we define first event spaces over finite types of outcomes and 
then we give a type definition for probability distribution. 

The possible outcomes can be provided as a type represented here by a type variable $\Omega$.
This type is assumed to be finite implicitly by coercing the type variable $\Omega$ into the
type class \texttt{finite} using the type judgment with \texttt{::} in the following definition
of probability space.
\begin{ttbox}
{\bf definition} prob_space :: ((\ttOmega :: finite) set) set) \ttfun bool
{\bf where} prob_space \ttevs = \{\} \ttin \ttevs \ttand (UNIV :: \ttOmega set) \ttin \ttevs \ttand
                       \ttforall A, B \ttin \ttevs. A \ttcup B \ttin \ttevs \ttand
                       \ttforall A \ttin \ttevs. (UNIV :: \ttOmega set) - A \ttin \ttevs 
\end{ttbox}
In the above type definition, the $\Omega$ is an Isabelle type variable. The 
polymorphic constructor \texttt{UNIV} is a standard constructor in Isabelle and 
represents the set of all elements of a type, here all outcomes in $\Omega$.
We can now show that the powerset over a finite type is a probability space. 
\begin{ttbox}
{\bf theorem} Pow_prob_space: "\ttforall (A :: (\ttOmega::finite) set). prob_space (Pow A)"
\end{ttbox}
A probability distribution is a function over a probability space. We use a type definition
for it.
\begin{ttbox}
{\bf typedef} (\ttOmega :: finite) prob_dist = \{p :: (\ttOmega set \ttfun real). 
                        \ttforall (A :: \ttOmega set). p A \ttgeq 0 \ttand p(UNIV :: \ttOmega set) = 1 \ttand 
                        \ttforall (A :: \ttOmega set) B . A \ttcap B = \ttvarnothing \ttimp p(A \ttcup B) = p(A) + p(B) \}
\end{ttbox}
In the above type definition for probability distribution, we can see that the 
three criteria from Definition \ref{def:dist} are almost one to one translated into Isabelle.
%However, since we defined a type of probability spaces, we need to lift the usual
%set operators to this new type. These operators are indicated by an index $\Omega$.
%We omit their technical definitions (see Appendix). 
Type definitions are applied by imposing them on new constants or variables which
automatically leads to the invocation of the defining properties on these elements: 
either by assuming them for constants defined over the new types or by 
creating new proof obligations when existing terms are judged to be of these types.
We apply this when we define a probability distribution over the power set of
a finite type of outcomes for QKD in Section \ref{sec:sec}.

Hurd already writes ``Measure theory defines what probability spaces are but does little
to help us find concrete distributions''\cite{hur:03}. He then uses Caratheodory's extension
theorem to help out. For the simple case of finite
sets of outcomes that we consider here, we introduce a canonical construction that 
uses the powerset of outcomes as the event space and accordingly constructs
the probability distribution by summing up the probabilities for the 
individual outcomes of any subset of $\Omega$, i.e. an event $\in \evs$, which is possible
since they are finite sets and the outcomes are all distinct.
For the definition of a generic operator for this canonical construction, we use
the \texttt{fold} operator available in Isabelle for defining simple recursive functions
over finite sets. Intuitively, \texttt{fold} operates like this:
\[ \text{\texttt{fold}}\ f\ z\ \{x_1, \dots, x_n\} = f(x_1 \dots (f\ x_n\ z))\,. \]

We define the canonical construction for probability distributions as a 
function \texttt{pmap} lifting a probability assignment \texttt{ops} for single
outcomes $\ttin \Omega$ to any set \texttt{S}.
\begin{ttbox}
pmap (ops :: \ttOmega \ttfun real) S = fold (\ttlam x y. ops x + y) 0 S
\end{ttbox}
Now, we can show that for any finite type $\Omega :: \texttt{finite}$ with a 
probability assignment \texttt{ops} the canonical construction yields a
probability distribution over the power set by showing that it is contained
in the defining set of the type \texttt{prob\_dist} given by 
the domain of the internal type injection \texttt{Rep\_prob\_dist}.
\begin{ttbox}
{\bf theorem} pmap_ops:
  pmap ops \ttin dom Rep_prob_dist \ttand \ttforall x :: \ttOmega. pmap ops \{x\} = ops x
\end{ttbox}

Conditional probability, for example, $P(A | B)$ signifies the probability for an 
event $A$ given an event $B$. It can be defined simply as follows.
\begin{definition}[Conditional Probability]\label{def:condprob}
For an event space $\evs$ and two events $A, B \in \evs$ the conditional probability of $A$ given
$B$ is defined for a probability distribution $P$ as
\[P(A | B) \equiv P(A \cap B) / P(B)\,.\] 
\end{definition}
The corresponding Isabelle definition uses some syntactic sugaring to hide the fact that 
the mathematical definition above is somewhat sloppy in its types.
\begin{ttbox}
{\bf definition} cond_prob ::  (\ttOmega :: finite)prob_dist \ttfun \ttOmega set \ttfun \ttOmega set \ttfun real "_(_|_)" 
{\bf where}   P(A|B) \ttequiv (Rep_prob_dist P (A \ttcap B)) / (Rep_prob_dist P B)
\end{ttbox}
The above Isabelle definition uses the mixfix syntax after the type in quotation marks to 
allow writing the same probability distribution as a function with two arguments by a syntactic
translation into the corresponding definition with intersections of the event sets. 
%The numbers at the end of the first line define the parser precedence of the syntactic operator.

\section{Protocol Tree with Probabilities}

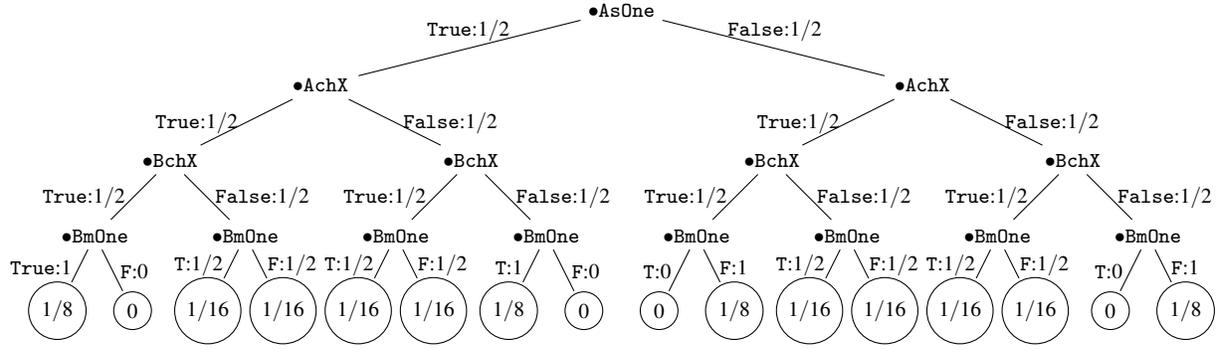
\begin{figure*}%[ht]
\begin{center}
\begin{tikzpicture}[>=stealth',shorten >=1pt,auto,scale=0.5]
\scriptsize
	\node  (r) at (15,8) {$\bullet$\texttt{AsOne}};
        \node  (o1) at (7,6) {$\bullet$\texttt{AchX}};
        \node  (o2) at (23,6) {$\bullet$\texttt{AchX}};
        \node  (t1) at (3,4) {$\bullet$\texttt{BchX}};
        \node  (t2) at (11,4) {$\bullet$\texttt{BchX}};
        \node  (t3) at (19,4) {$\bullet$\texttt{BchX}};
        \node  (t4) at (27,4) {$\bullet$\texttt{BchX}};
        \node  (r1) at (1,2)  {$\bullet$\texttt{BmOne}};
        \node  (r2) at (5,2)  {$\bullet$\texttt{BmOne}};
        \node  (r3) at (9,2)  {$\bullet$\texttt{BmOne}};
        \node  (r4) at (13,2)  {$\bullet$\texttt{BmOne}};
        \node  (r5) at (17,2)  {$\bullet$\texttt{BmOne}};
        \node  (r6) at (21,2)  {$\bullet$\texttt{BmOne}};
        \node  (r7) at (25,2)  {$\bullet$\texttt{BmOne}};
        \node  (r8) at (29,2)  {$\bullet$\texttt{BmOne}};
        \node  (f1) at (0,0)[circle,draw]   {$1/8$};
        \node  (f2) at (2,0)[circle,draw]   {$0$};
        \node  (f3) at (4,0)[circle,draw]   {$1/16$};
        \node  (f4) at (6,0)[circle,draw]   {$1/16$};
        \node  (f5) at (8,0)[circle,draw]   {$1/16$};
        \node  (f6) at (10,0)[circle,draw]   {$1/16$};
        \node  (f7) at (12,0)[circle,draw]   {$1/8$};
        \node  (f8) at (14,0)[circle,draw]   {$0$};
        \node  (f9) at (16,0)[circle,draw]   {$0$};
        \node  (f10) at (18,0)[circle,draw]   {$1/8$};
        \node  (f11) at (20,0)[circle,draw]   {$1/16$};
        \node  (f12) at (22,0)[circle,draw]   {$1/16$};
        \node  (f13) at (24,0)[circle,draw]   {$1/16$};
        \node  (f14) at (26,0)[circle,draw]   {$1/16$};
        \node  (f15) at (28,0)[circle,draw]   {$0$};
        \node  (f16) at (30,0)[circle,draw]   {$1/8$};

        \draw (r) to node[above]{\texttt{True}:$1/2$} (o1);
        \draw (r) to node[above]{\texttt{False}:$1/2$} (o2);
        \draw (o1) to node[left]{\texttt{True}:$1/2$} (t1);
        \draw (o1) to node[right]{\texttt{False}:$1/2$} (t2);
        \draw (o2) to node[left]{\texttt{True}:$1/2$} (t3);
        \draw (o2) to node[right]{\texttt{False}:$1/2$} (t4);
        \draw (t1) to node[left]{\texttt{True}:$1/2$} (r1);
        \draw (t1) to node[right]{\texttt{False}:$1/2$} (r2);
        \draw (t2) to node[left]{\texttt{True}:$1/2$} (r3);
        \draw (t2) to node[right]{\texttt{False}:$1/2$} (r4);
        \draw (t3) to node[left]{\texttt{True}:$1/2$} (r5);
        \draw (t3) to node[right]{\texttt{False}:$1/2$} (r6);
        \draw (t4) to node[left]{\texttt{True}:$1/2$} (r7);
        \draw (t4) to node[right]{\texttt{False}:$1/2$} (r8);
        \draw (r1) to node[left]{\texttt{True}:$1$} (f1);
        \draw (r1) to node[right]{\texttt{F}:$0$} (f2);
        \draw (r2) to node[left]{\texttt{T}:$1/2$} (f3);
        \draw (r2) to node[right]{\texttt{F}:$1/2$} (f4);
        \draw (r3) to node[left]{\texttt{T}:$1/2$} (f5);
        \draw (r3) to node[right]{\texttt{F}:$1/2$} (f6);
        \draw (r4) to node[left]{\texttt{T}:$1$} (f7);
        \draw (r4) to node[right]{\texttt{F}:$0$} (f8);
        \draw (r5) to node[left]{\texttt{T}:$0$} (f9);
        \draw (r5) to node[right]{\texttt{F}:$1$} (f10);
        \draw (r6) to node[left]{\texttt{T}:$1/2$} (f11);
        \draw (r6) to node[right]{\texttt{F}:$1/2$} (f12);
        \draw (r7) to node[left]{\texttt{T}:$1/2$} (f13);
        \draw (r7) to node[right]{\texttt{F}:$1/2$} (f14);
        \draw (r8) to node[left]{\texttt{T}:$0$} (f15);
        \draw (r8) to node[right]{\texttt{F}:$1$} (f16);
\end{tikzpicture}
\end{center}\label{fig:tree}
\caption{Decision tree for QKD: probabilities for each step at edges and 
outcome probabilities as leaves.}
\end{figure*}
The tree depicted in Figure \ref{fig:tree} shows the probabilities along the
paths from root to leaves according to the steps of the protocol.

How does a concrete application, like QKD, relate to the probabilities and distributions
encountered previously?
The outcomes we consider are Boolean vectors of length four representing each one 
a possible path of the protocol. The following type of \texttt{QKD\_om} will
instantiate the polymorphic type $\Omega$ in the previous definitions.
\begin{ttbox}
type_synonym QKD_om = bool * bool * bool * bool
\end{ttbox}
To make the outcomes more readable we introduce the following abbreviations
as local definition of the locale QKD \cite{kam:00}. Their intended meaning is
\texttt{AsOne}: ``A sends 1'', \texttt{AchX/BchX}: ``A/B chooses diagonal scheme''
and \texttt{BmOne}'' ``B measures 1''.
\begin{ttbox}
{\bf locale} QKD = 
{\bf defines}
AsOne out = fst(out)
AchX out = fst(snd out)
BchX out = fst(snd(snd out))
BmOne out = snd(snd(snd out))
\end{ttbox}
The basic distribution on these 16 outcomes derived from Figure \ref{fig:tree}
is given in the table in Figure \ref{fig:table}.
\begin{figure}
\begin{center}
\begin{tabular}{|l|l|l|l|l|}
\hline
AsOne & AchX & BchX & BmOne & P \\
\hline
False & False & False & False & 1/8 \\
False & False & False & True & 0 \\
False & False & True & False & 1/16 \\
False & False & True & True & 1/16 \\
False & True  & False & False & 1/16 \\
False & True & False & True & 1/16 \\
False & True & True & False & 1/8 \\
False & True & True & True & 0 \\
True & False & False & False & 0 \\
True & False & False & True & 1/8 \\
True & False & True & False & 1/16 \\
True & False & True & True & 1/16 \\
True & True & False & False & 1/16 \\
True & True & False & True & 1/16 \\
True & True & True & False & 0 \\
True & True & True & True & 1/8 \\
\hline
\end{tabular}
\caption{Probability assignment for QKD outcomes}\label{fig:table}
\end{center}
\end{figure}
This basic probability distribution can be input as the element \texttt{ops} 
in the above function \texttt{pmap} producing automatically the canonical probability 
distribution for the QKD protocol.
To define the basic function \texttt{ops}, we use a locale definition \cite{kam:00}. 
We omit the details of the cases as they are clear from the table in Figure \ref{fig:table}
(see Appendix).
\begin{ttbox}
{\bf defines} (qkd_ops  :: QKD_om \ttfun real) = 
      \ttlam x. case x of 
            (False, False, False, False) \ttfun 1/8 
          | ...
\end{ttbox}
We can then define a probability distribution being able to show that it is in fact
one by Theorem \texttt{pmap\_ops}.
\begin{ttbox}
{\bf defines} (qkd_prob_dist  :: prob_dist) = pmap qkd_ops
\end{ttbox}
Based on this probability distribution we can calculate interesting probabilities
telling us something about the security of the protocol.

In order to do that we first consider another useful probability law: the law of total
probability.

\section{Law of Total Probability}
For the Security Argument, we need the law of total probability.
\begin{theorem}[Law of total probability]
Let $A_j, j \leq n$ for some $n \in \mathbb{N}$ be a set of events 
partitioning the event space $\evs$, that is, $\forall\ i,j \leq n.\ i \neq j \ttfun A_i \cap A_j = \varnothing$ and $\bigcup_j A_j = \Omega$. Let further $B \in evs$. We then have that
\[P(B) = \sum_{j} P(B | A_j) P(A_j)\,.\]
\end{theorem}
\begin{proof}
Since we have a partition, that is, $A_i \cap A_j = \varnothing$ for all $i,j \leq n$ with 
$i \neq j$, we have also 
\[
\begin{array}{ccccccccr}
(B \cap A_i) \cap (B \cap A_j) & =  & B \cap (A_i \cap A_j) & =  & B \cap \varnothing & = & \varnothing & \qquad\qquad (1)
\end{array}
\]
Therefore
\[
\begin{array}{clll}
P(B) & = & P(B \cap \Omega) & {(A_j\ \text{is partition of } \Omega)}\\
& = & P(B \cap (A_1 \cup \dots \cup A_n)) & (\text{set algebra})\\
& = & P((B \cap A_1) \cup \dots \cup (B \cap A_n)) & ((1) \text{ and Definition \ref{def:dist} (3)})\\
& = & P(B \cap A_1) + \dots + P(B \cap A_n) & (\text{summation})\\
& = & \sum_{j} P(B \cap A_j) & (\text{Definition of conditional probability}) \\
& = & \sum_{j} P(B | A_j)* P(A_j) &  \\
\end{array}
\]
\end{proof}

\section{Security Argument}
\label{sec:sec}
The first argument computes the probability that B measures 1 applying the law of 
total probability.
The partition ${\mathcal A}$ of $\Omega$ used in the derivation is given as the following family 
of disjoint sets with $\bigcup {\mathcal A} = \Omega$.
\[
\begin{array}{lll}
{\mathcal A} & = & \{ \{ s :: \texttt{QKD\_om}.\ \texttt{BchX}\ s \wedge \texttt{AchX}\ s \wedge \texttt{AsOne}\ s \},\\
& & \ \{ s :: \texttt{QKD\_om}.\ \neg(\texttt{BchX}\ s) \wedge \texttt{AchX}\ s \wedge \texttt{AsOne}\ s \},\\
& & \ \{ s :: \texttt{QKD\_om}.\ \texttt{BchX}\ s \wedge \neg(\texttt{AchX}\ s) \wedge \texttt{AsOne}\ s \}, \\
& & \ \{ s :: \texttt{QKD\_om}.\ \neg(\texttt{BchX})\ s \wedge \neg(\texttt{AchX}\ s) \wedge \texttt{AsOne}\ s \}, \\
& & \ \{ s :: \texttt{QKD\_om}.\ \texttt{BchX}\ s \wedge \texttt{AchX}\ s \wedge \neg(\texttt{AsOne})\ s \},\\
& & \ \{ s :: \texttt{QKD\_om}.\ \neg(\texttt{BchX}\ s) \wedge \texttt{AchX}\ s \wedge \neg(\texttt{AsOne})\ s \},\\
& & \ \{ s :: \texttt{QKD\_om}.\ \texttt{BchX}\ s \wedge \neg(\texttt{AchX}\ s) \wedge \neg(\texttt{AsOne})\ s \}, \\
& & \ \{ s :: \texttt{QKD\_om}.\ \neg(\texttt{BchX})\ s \wedge \neg(\texttt{AchX}\ s) \wedge \neg(\texttt{AsOne})\ s \} \\
& & \}
\end{array}
\]
For each $A_j \in {\mathcal A}$, we have $P(A_j) = 1/8$:
since $P$ is a probability distribution, we can use the third defining property
to sum up the disjoint probabilities for each outcome. The outcome probabilities in 
Figure \ref{fig:table} give for example (similar for the other $A_j$):
\[ 
\begin{array}{ll}
P(\{ s :: \texttt{QKD\_om}.\ \texttt{BchX}\ s \wedge \texttt{AchX}\ s \wedge \texttt{AsOne}\ s \}) & = \\
    P(\texttt{True}, \texttt{True},\texttt{True}, \texttt{True}) +
    P(\texttt{True}, \texttt{True},\texttt{True}, \texttt{False}) & = \\
    1/8 + 0 = 1/8 &  \\
\end{array}
\]
With this we can compute that $P(\texttt{BmOne}) = 1/2$.\footnote{This calculation can
be much simplifed if we apply the second to last step of total probability instead but
the additional steps are instructive.} 
\[
\begin{array}{clll}
P(\texttt{BmOne}) & = & \sum_{A_j \in \mathcal{A}} P(\texttt{BmOne} | A_j) * P(A_j) &  (\text{Law of total probability})\\
 & = & 1/8 * \sum_{A_j \in \mathcal{A}} P(\texttt{BmOne} | A_j) & (P(A_j) = 1/8)\\
 & = & 1/8 * \sum_{A_j \in \mathcal{A}} P(\texttt{BmOne} \cap A_j) / P(A_j) & (\texttt{Definition } \ref{def:condprob})\\
 & = & \sum_{A_j \in \mathcal{A}} P(\texttt{BmOne} \cap A_j) & (P(A_j) = 1/8)\\
 & = & 1/2 & (\text{sum even columns in Table } \ref{fig:table})
\end{array}
\]

This probability cannot be interpreted as a security statement directly. 
It rather says that on the whole Bob receives 1s with 50\% probability 
but not how this relates to what A has actually sent.
However the above probability $P(\texttt{BmOne})$ is useful to calculate the
conditional probability $P(\texttt{AsOne} | \texttt{BmOne})$: how likely is it
that A has actually sent a 1 given that B received a 1?
\[
\begin{array}{clll}
P(\texttt{AsOne} | \texttt{BmOne}) & = & P(\texttt{AsOne} \cap \texttt{BmOne}) / P(\texttt{BmOne})
                                    &  (\text{Conditional Probability})\\
& = & 2*(P(\{ x:: \ttOmega.\ \texttt{AsOne } x \wedge \texttt{BmOne } x \})) 
                                    & (\text{above calculation}) \\
& = & 2*(P(\{\{\texttt{T,F,F,T}\},\{\texttt{T,F,T,T}\}, & \\
&   & \qquad\quad \{\texttt{T,T,F,T}\}, \{\texttt{T,T,T,T}\}\}))    & (\text{Definition \ref{def:dist}}) \\
& = & 2*(P(\{\{\texttt{T,F,F,T}\}) + P(\{\texttt{T,F,T,T}\}) + & \\
&   & \qquad\quad P(\{\texttt{T,T,F,T}\}) + P(\{\texttt{T,T,T,T}\}))   & (\text{Table \ref{fig:table}}) \\
& = & 2*(1/8 + 1/16 + 1/16 + 1/8) & (\texttt{arithmetic})\\
& = & 3/4  \\
\end{array}
\]
This shows that there is 25\% chance of error for Bob and Eve to receive 
the wrong bit. 

\section{Conclusion}
We have provided a simple probability model for Isabelle based on finite outcome
types formalising basic Bayesian probability notions, proving general theorems
and illustrating their use on the application of QKD. It is interesting to observe
that the mathematical model for the quantum computations realising the protocol behaviour
seems not present: it is embedded in the a priori probabilities of the basic
outcomes in Table \ref{fig:table}. This is a nice (and important) observation since the present
formalisation shows that the actual quantum model and the probability model can 
be treated in a modular manner.

The security argument presented so far is actually rather poor: with 75\%
probability the attacker Eve and Bob can assume that they have got the right bit. 
This is not enough protection for security.
The security relies on one further observation and an additional protocol step.
The observation is that if Eve resends the bit she has received, Bob -- when
measuring afterwards -- will have erroneous measurements even in the case that should
give the right results with 100\% probability. This error can be revealed in the 
additional protocol step in which Alice and Bob compare (first the used polarisation 
schemes and then based on that) some portion of the bits they transmitted coincidentally 
with equal schemes.

To model and analyse the probabilities for the intercept and resend attack, the 
finite probability model presented her is sufficient: we just need to extend the 
outcome type and add the corresponding a priori probabilities to an extended 
Table \ref{fig:table}.
This will allow to calculate the final error probability when Eve intercepts and resends.
However, for the final security argument that proves the ``unconditional security'' of 
QKD (assuming the clear line isn't intercepted) the probability model over finite
outcome types needs to be extended to outcomes that are infinite sequences. This will then
necessitate a model similar to the one in Hurd's thesis \cite{hur:03}.

% Related work : \cite{kam:18b, kam:18a}

\bibliographystyle{abbrv}
\bibliography{biblio}

\appendix
\section{Isabelle Code}
under construction
\end{document}